\documentclass[letterpaper,12pt]{article}
\usepackage{amssymb, amsmath}
\usepackage{amsthm}
\usepackage{algorithm}
\usepackage{algorithmic}
\usepackage{enumerate}
\usepackage[numbers,square,comma,sort&compress]{natbib}
\usepackage[pdftex,colorlinks]{hyperref}
\newtheorem{theorem}{Theorem}
\newtheorem{fact}[theorem]{Fact}
\newtheorem{lemma}[theorem]{Lemma}

\newcommand{\comment}[1]{}

\begin{document}
%\title{Efficient approximations of closeness centrality}
% \footnote{The authors
%are supported in part by the National Science Council of Taiwan
%under grant 97-2221-E-002-096-MY3 and Excellent Research Projects of
%National Taiwan University under grant
%%97R0062-05.
%98R0062-05.
%}}
%\title{Sublinear bounds for $1$-median selection in metric spaces}
\title{On ultrametric $1$-median selection}
%\title{Metric $1$-median selection via a new pseudo-metric}
%\title{Deterministic sublinear-time algorithms for $1$-median selection}
%\text{http://par.cse.nsysu.edu.tw/$\sim$algo/paper/paper13/10.pdf}}}
%\title{On the approximability of metric $1$-median selection by nonevasive algorithms}
%\title{Metric $1$-median selection revisited}

\author{
Ching-Lueh Chang \footnote{Department of Computer Science and
Engineering,
%\& Innovation Center for Big Data and Digital Convergence,
Yuan Ze University, Taoyuan, Taiwan. Email:
clchang@saturn.yzu.edu.tw}
%\footnote{Innovation Center for Big Data and Digital Convergence, Yuan Ze University, Taoyuan, Taiwan.}
%\footnote{Supported in part by the Ministry of Science and Technology
%of Taiwan under
%grant
%101-2221-E-155-015-MY2.}
%103-2221-E-155-026-MY2.}
}

%\institute{
%Department of Computer Science and Information Engineering, National Taiwan
%University, Taipei, Taiwan.\\
%\email{d95007@csie.ntu.edu.tw}
%\and
%Department of Computer Science and Information Engineering, National Taiwan
%University, Taipei, Taiwan.\\
%\email{lyuu@csie.ntu.edu.tw}
%}

%\author{Ching-Lueh Chang\thanks{Department of Computer Science and
%Information Engineering, National Taiwan University, Taipei,
%Taiwan. Email: d95007@csie.ntu.edu.tw.} \and Yuh-Dauh Lyuu\thanks{Department
%of Computer Science and Information
%Engineering, National Taiwan University, Taipei,
%Taiwan. Email: lyuu@csie.ntu.edu.tw.}}
\maketitle

\begin{abstract}
Consider the problem of finding a point in an ultrametric space
with the minimum average distance to all points.
We give this problem a Monte Carlo $O((\log^2(1/\epsilon))/\epsilon^3)$-time
$(1+\epsilon)$-approximation algorithm for all $\epsilon>0$.
\end{abstract}

\section{Introduction}

A metric space is a
nonempty set $M$ endowed with a distance function $d\colon M\times M\to[0,\infty)$
%pair $(M,d\colon M\times M\to[0,\infty))$
%such that
satisfying
\begin{itemize}
\item $d(x,y)=0$ if and only if $x=y$,
\item $d(x,y)=d(y,x)$, and
\item $d(x,z)\le d(x,y)+d(y,z)$ (triangle inequality)
\end{itemize}
for all $x$, $y$, $z\in M$.
With the triangle inequality strengthened to
%the {\em ultrametric}
$$
d\left(x,z\right)\le \max\left\{d\left(x,y\right),\,d\left(y,z\right)\right\},
$$
we call
$(M,d)$
%is called
an ultrametric space and $d$ an ultrametric (a.k.a.\
non-Archimedean metric or super-metric).
The mathematical community studies ultrametrics extensively.
%This paper studies sublinear-time computation regarding ultrametric spaces,
%where ``sublinear''

Given an $n$-point metric space $(M,d)$,
{\sc metric $1$-median} asks for a point in $M$, called
a $1$-median, with the minimum average distance
to all points.
{\sc Metric $1$-median} is a special case of the classical $k$-median
%problem
clustering
and a generalization to the classical median selection~\cite{CLRS09}.
It can also be interpreted as finding the most important point because social
network analysis often measures the importance of an actor $v$ by $v$'s closeness
centrality, defined to be $v$'s average
distance to all points~\cite{WF94}.
Not surprisingly, {\sc metric $1$-median} is extensively studied, e.g.,
%For example, Kumar et al.~\cite{KSS10} give $O(1)$-time
in the general~\cite{Ind99,Ind00}, Euclidean~\cite{KSS10}, streaming~\cite{GMMMO03}
%, online~\cite{MP04}
and deterministic~\cite{Cha18}
cases.
Indyk~\cite{Ind99,Ind00} has the currently best upper bound for {\sc metric $1$-median}:

\begin{theorem}[\cite{Ind99,Ind00}]\label{greattheoremofIndyk}
{\sc Metric $1$-median} has a Monte Carlo $O(n/\epsilon^2)$-time
$(1+\epsilon)$-approximation algorithm for all $\epsilon>0$.
\end{theorem}

The greatest
%advantage
strengths
of Theorem~\ref{greattheoremofIndyk} are the {\em sublinear}
%running time,
time complexity (of $O(n/\epsilon^2)$) and the optimal approximation ratio (of $1+\epsilon$),
where ``sublinear'' means ``$o(n^2)$'' by convention because there are $\Theta(n^2)$
distances.
Furthermore, except
%Except
for the dependence of the time complexity on $\epsilon$,
all parameters in Theorem~\ref{greattheoremofIndyk} are
%known
easily shown
to be
optimal~\cite[Sec.~7]{Cha12}.
%However,

Chang~\cite[Sec.~6]{Cha12} uses Indyk's~\cite[Sec.~6.1]{Ind00} technique to give
%There is also
a Monte Carlo
%$(2+\epsilon)$-approximation
algorithm for {\sc metric $1$-median}
with time complexity
{\em independent} of $n$ but at the cost of a worse approximation ratio:
%Using Indyk's~\cite{Ind00} technique, Chang~\cite{Cha12}
%proves the following.

\begin{theorem}[{\cite[Sec.~6]{Cha12}}]\label{Changconstanttime}
For all $\epsilon>0$,
{\sc metric $1$-median} has a Monte Carlo $O((\log^2 (1/\epsilon))/\epsilon^3)$-time
$(2+\epsilon)$-approximation algorithm
%for all $\epsilon>0$.
with success probability greater than $1-\epsilon$.
\end{theorem}

Let {\sc ultrametric $1$-median} be {\sc metric $1$-median} restricted to
ultrametric spaces.
%Although the
The
approximation ratio of $2+\epsilon$ in Theorem~\ref{Changconstanttime}
cannot be improved to $2-\epsilon$ even if we require the success
probability only to be a small constant~{\cite[Sec.~7]{Cha12}}.
In contrast,
this paper gives
a Monte Carlo $O((\log^2 (1/\epsilon))/\epsilon^3)$-time
$(1+\epsilon)$-approximation algorithm for {\sc ultrametric $1$-median}.
%for all $\epsilon>0$.
So our algorithm has the optimal approximation ratio (of $1+\epsilon$) and a time complexity
(of $O((\log^2 (1/\epsilon))/\epsilon^3)$)
independent of $n$.

\section{Algorithm}

For all $n\in\mathbb{Z}^+$,
$[n]\stackrel{\text{def.}}{=}\{1,2,\ldots,n\}$
by convention.
Let $([n],d)$ be an ultrametric space, OPT a $1$-median of $([n],d)$
and
%$\alpha\in[n]$,
%$\epsilon\in(0,1)$.
$\epsilon>0$.
%and $C\stackrel{\text{def.}}{=}(-1+\sqrt{5})/2=0.618\ldots$
%For all $p\in [n]$ and $Q\subseteq[n]$, let
%$$S_Q(p)\stackrel{\text{def.}}{=}\sum_{q\in Q}\,d\left(p,q\right)$$
%be the total distance from $p$ to the points in $Q$,
%following Indyk~\cite[Sec.~6.1]{Ind00}.
Order the points in
%$[n]\setminus\{\text{OPT}\}$
$[n]$
as
$p_1=\text{\rm OPT}$,
$p_2$,
%$p_2$,
$\ldots$,
%$p_{n-1}$
$p_n$
so that
\begin{eqnarray}
0=d\left(\text{OPT},p_1\right)
\le d\left(\text{OPT},p_2\right)
\le \cdots
\le
%d\left(\text{OPT},p_{n-1}\right).
d\left(\text{OPT},p_n\right).
\label{orderofincreasingdistances}
\end{eqnarray}
Furthermore,
let
\begin{eqnarray}
r^*\stackrel{\text{def.}}{=}\frac{1}{n}\cdot
%S_{[n]}
\sum_{i=1}^n\,
d\left(\text{\rm OPT},p_i\right)
\label{closenesscentralityofthe1median}
\end{eqnarray}
be the average distance from a $1$-median to all points.
%and
%$r^*\stackrel{\text{def.}}{=} \sum_{p=1}^n\,d(\text{OPT},p)/n$,
%.
%and $\alpha\in[n]$.
%such that
%$\alpha\in B(\text{OPT},(1+\epsilon)r^*)$.
%Furthermore,
%$r\stackrel{\text{def.}}{=} \sum_{i=1}^n\,d(\alpha,i)/n$.
Because the brute-force algorithm for {\sc ultrametric $1$-median}
takes $\Theta(n^2)$ time and we want an
$O((\log^2 (1/\epsilon))/\epsilon^3)$-time algorithm,
assume $\epsilon\ge n^{-2/3}$ W.L.O.G.
%We also assume W.L.O.G.\ that $\epsilon$ is less than a small constant.\footnote{Our
%result with, e.g., $\epsilon=0.001$}
Furthermore, assume $\epsilon\le 0.0001$ W.L.O.G.\footnote{It is easy to see that
if
%the case of $\epsilon=0.0001$ of
our result
%is true
holds
%for
%the case
%that
%all
when
$\epsilon=0.0001$, then it also holds for all $\epsilon> 0.0001$.}
\comment{ % not used??? 20190901
The following fact is well-known.

\begin{fact}\label{isoscelesproperty}
For all $x$, $y$, $z\in[n]$ satisfying $d(x,y)\le d(x,z)$,
$$d\left(x,y\right)\le d\left(y,z\right)=d\left(x,z\right).$$
\end{fact}
} % not used??? 20190901

\comment{ % writing everything within 1 lemma might make that lemma seem hard?
\begin{lemma}\label{lowerboundforoptimal}
For all
%$\ell\in[n-1]$,
$1\le\ell\le n$,
$$
\sum_{j=1}^n\,d\left(\text{\rm OPT},p_j\right)
\ge \left(n-\ell+1\right)\cdot d\left(\text{\rm OPT},p_\ell\right)
$$
\end{lemma}
\begin{proof}
Clearly,
$$
\sum_{j=1}^n\,d\left(\text{\rm OPT},p_j\right)
\ge \sum_{j=\ell}^{n}\,d\left(\text{\rm OPT},p_j\right)
\stackrel{\text{(\ref{orderofincreasingdistances})}}{\ge}
\sum_{j=\ell}^{n}\,d\left(\text{\rm OPT},p_\ell\right)
%\left(n-\ell+1\right)\cdot d\left(\text{\rm OPT},p_\ell\right).
$$
\end{proof}
}% writing everything within 1 lemma might make that lemma seem hard?

\comment{ % writing tiny lemmas might make the article look even routinely done?
\begin{lemma}\label{tocloserpoints}
For all
%$\ell\in[n-1]$ and $i\in[\ell-1]$,
$1\le i<\ell\le n$,
$$
d\left(p_\ell,p_i\right)\le d\left(\text{\rm OPT},p_\ell\right).
$$
\end{lemma}
\begin{proof}
We have
$$
d\left(p_\ell,p_i\right)\le
\max\left\{d\left(\text{OPT},p_\ell\right),\,
d\left(\text{OPT},p_i\right)
\right\}
\stackrel{\text{(\ref{orderofincreasingdistances})}}{=}
d\left(\text{OPT},p_\ell\right).
$$
\end{proof}

\begin{lemma}\label{tofurtherpoints}
For all
%$\ell\in[n-1]$ and $i\in[n]\setminus[\ell]$,
$1\le \ell <i\le n$,
$$
d\left(p_\ell,p_i\right)=d\left(\text{\rm OPT},p_i\right).
$$
\end{lemma}
\begin{proof}
We have
$$
d\left(p_\ell,p_i\right)\le
\max\left\{d\left(\text{OPT},p_\ell\right),\,
d\left(\text{OPT},p_i\right)
\right\}
\stackrel{\text{(\ref{orderofincreasingdistances})}}{=}
d\left(\text{OPT},p_i\right).
$$
\end{proof}
}% writing tiny lemmas might make the article look even routinely done?

\begin{lemma}\label{nowthisshouldbethekeylemma}
For all
%$\ell\in[n-1]$,
$1\le\ell\le n$,
$$
\sum_{i=1}^n\,d\left(p_\ell,p_i\right)
\le
\left(1+\frac{\ell-1}{n-\ell+1}\right)
\sum_{i=1}^n\,d\left(\text{\rm OPT},p_i\right).
$$
\end{lemma}
\begin{proof}
We have
\begin{eqnarray*}
&&\sum_{i=1}^n\,d\left(p_\ell,p_i\right)\\
&=&\sum_{i=1}^{\ell-1}\,d\left(p_\ell,p_i\right)
+\sum_{i=\ell+1}^{n}\,d\left(p_\ell,p_i\right)\\
&\le&\sum_{i=1}^{\ell-1}\,\max\left\{d\left(\text{OPT},p_\ell\right),\,
d\left(\text{OPT},p_i\right)\right\}
+\sum_{i=\ell+1}^{n}\,\max\left\{d\left(\text{OPT},p_\ell\right),\,
d\left(\text{OPT},p_i\right)\right\}\\
&\stackrel{\text{(\ref{orderofincreasingdistances})}}{\le}&
%&\stackrel{\text{Lemmas~\ref{tocloserpoints}--\ref{tofurtherpoints}}}{\le}&
\sum_{i=1}^{\ell-1}\,d\left(\text{OPT},p_\ell\right)
+\sum_{i=\ell+1}^n\,d\left(\text{OPT},p_i\right)\\
&\le&
\sum_{i=1}^{\ell-1}\,d\left(\text{OPT},p_\ell\right)
+\sum_{i=1}^n\,d\left(\text{OPT},p_i\right)\\
&=&
\sum_{i=1}^{\ell-1}\,
\frac{1}{n-\ell+1}\cdot \sum_{j=\ell}^n\,d\left(\text{OPT},p_\ell\right)
+\sum_{i=1}^n\,d\left(\text{OPT},p_i\right)\\
&\stackrel{\text{(\ref{orderofincreasingdistances})}}{\le}&
\sum_{i=1}^{\ell-1}\,
\frac{1}{n-\ell+1}\cdot \sum_{j=\ell}^n\,d\left(\text{OPT},p_j\right)
+\sum_{i=1}^n\,d\left(\text{OPT},p_i\right)\\
%&\le&
%\sum_{i=1}^{\ell-1}\,
%\frac{1}{n-\ell+1}\cdot \sum_{j=\ell}^n\,d\left(\text{OPT},p_j\right)
%+\sum_{i=1}^n\,d\left(\text{OPT},p_i\right)\\
&\le&
\sum_{i=1}^{\ell-1}\,
\frac{1}{n-\ell+1}\cdot \sum_{j=1}^n\,d\left(\text{OPT},p_j\right)
+\sum_{i=1}^n\,d\left(\text{OPT},p_i\right)\\
&=&
\frac{\ell-1}{n-\ell+1}\cdot \sum_{i=1}^n\,d\left(\text{OPT},p_i\right)
+\sum_{i=1}^n\,d\left(\text{OPT},p_i\right).
\end{eqnarray*}
\end{proof}

In short, Lemma~\ref{nowthisshouldbethekeylemma}
says that $p_\ell$ is an approximate $1$-median for all small $\ell$.
%The following fact is due to Indyk~\cite{Ind00}.
Below is the key of the proof of Theorem~\ref{greattheoremofIndyk}.

\comment{ % written in pseudocode % 20190901
Pick ${\boldsymbol u}_1$, ${\boldsymbol u}_2$, $\ldots$, ${\boldsymbol u}_h$
independently and uniformly at random from $[n]$, where
$h\stackrel{\text{def.}}{=} 10^9(\log(1/\epsilon))/\epsilon$.
}% written in pseudocode % 20190901

\begin{fact}[{\cite[Sec.~6.1]{Ind00}}]\label{Indykkeyfact}
Pick ${\boldsymbol v}_1$, ${\boldsymbol v}_2$, $\ldots$,
${\boldsymbol v}_k$
independently and uniformly at random from $[n]$, where $k\in\mathbb{Z}^+$.
Then for
%For
all $a$, $b\in[n]$ satisfying
%$\sum_{p\in[n]}\,d(b,p)>(1+\epsilon)\,\sum_{p\in[n]}\,d(a,p)$,
$\sum_{j=1}^n\,d(b,p_j)>(1+\epsilon)\,\sum_{j=1}^n\,d(a,p_j)$,
$$
\Pr\left[
\sum_{j=1}^k\,d\left(b,{\boldsymbol v}_j\right)
\le
\sum_{j=1}^k\,d\left(a,{\boldsymbol v}_j\right)
%\ge
\right]
<\exp{\left(-\frac{\epsilon^2 k}{64}\right)}.
$$
\end{fact}

The
following
%technique of the following
%corollary, proved
lemma
uses
%a technique of
%is due to
Indyk's~\cite[Sec.~6.1]{Ind00}
technique
that Chang~\cite[Sec.~6]{Cha12} uses to prove
Theorem~\ref{Changconstanttime}.
%briefly
%for completeness, is implicit in~\cite[Sec.~6.1]{Ind00}.

\begin{lemma}
%[{Implicit in~{\cite[Sec.~6.1]{Ind00}}}]
\label{wemaytakethebestaccordingtothesamples}
%Given
Pick ${\boldsymbol v}_1$, ${\boldsymbol v}_2$, $\ldots$,
${\boldsymbol v}_k$ as in Fact~\ref{Indykkeyfact}, where $k=\lceil
10^9(\log(1/\epsilon))/\epsilon^2\rceil$.
Let
$x_1$, $x_2$, $\ldots$, $x_h\in[n]$,
where
%$h\stackrel{\text{def.}}{=} 10^9(\log(1/\epsilon))/\epsilon$,
$h=\lceil 10^9(\log(1/\epsilon))/\epsilon\rceil$,
and
%taking
\begin{eqnarray}
t=\mathop{\mathrm{argmin}}_{i=1}^h\,
\sum_{j=1}^k\,d\left(x_i,{\boldsymbol v}_j\right),
\label{thebestindexaccordingtorandomsamples}
\end{eqnarray}
breaking ties arbitrarily.
%an index $h\in [\ell]$ satisfying
Then
%$$\Pr\left[\sum_{p\in[n]}\,d\left(x_t,p\right)
%\le \left(1+\epsilon\right)\cdot \min_{i=1}^h\,
%\sum_{p\in[n]}\,d\left(x_i,p\right)\right]>1-\epsilon.$$
$$\Pr\left[\sum_{j=1}^n\,d\left(x_t,p_j\right)
\le \left(1+\epsilon\right)\cdot \min_{i=1}^h\,
\sum_{j=1}^n\,d\left(x_i,p_j\right)\right]>1-\epsilon.$$
%can be found
%with probability at least $1-\epsilon$.
%, using
%Furthermore, $t$ can be found
%in
%$O((\log(1/\epsilon))\ell/\epsilon^2)$
%queries to $d$.
%time.
\end{lemma}
\begin{proof}
%[Sketch of proof.]
Let
%$$
\begin{eqnarray}
%$$
%breaking ties arbitrarily.
%Furthermore,
%denote
%let
%$$
i^*&=&\mathop{\mathrm{argmin}}_{i=1}^h\,
%\sum_{p\in[n]}\, d\left(x_i,p\right),
%S_{[n]}
%\sum_{p\in[n]}\,d\left(x_i,p\right),\label{thebestfromthesamples}
\sum_{j=1}^n\,d\left(x_i,p_j\right),\label{thebestfromthesamples}
%$$
\end{eqnarray}
breaking ties arbitrarily.
%for the choice of $i^*$.
Then
{\footnotesize % inequalities too long 20190901
\begin{eqnarray*}
&&\Pr\left[
%S_{[n]}
\sum_{j=1}^n\,
d\left(x_t,p_j\right)> \left(1+\epsilon\right)\cdot \min_{i=1}^h\,
%S_{[n]}
\sum_{j=1}^n\,
d\left(x_i,p_j\right)\right]\\
&\stackrel{\text{(\ref{thebestfromthesamples})}}{=}&\Pr\left[
%S_{[n]}
\sum_{j=1}^n\,
d\left(x_t,p_j\right)>\left(1+\epsilon\right)\cdot
%S_{[n]}
\sum_{j=1}^n\,
d\left(x_{i^*},p_j
\right)\right]\\
&\stackrel{\text{(\ref{thebestindexaccordingtorandomsamples})}}{=}&
\Pr\left[
\left(
%S_{[n]}
\sum_{j=1}^n\,
d\left(x_t,p_j\right)>\left(1+\epsilon\right)\cdot
%S_{[n]}
\sum_{j=1}^n\,
d\left(x_{i^*},p_j\right)
\right)\land\left(
\sum_{j=1}^k\,d\left(x_t,{\boldsymbol v}_j\right)
=
\min_{i=1}^h\,\sum_{j=1}^k\,d\left(x_i,{\boldsymbol v}_j\right)\right)
\right]\\
&\le&
\Pr\left[
\left(
%S_{[n]}
\sum_{j=1}^n\,
d\left(x_t,p_j\right)>\left(1+\epsilon\right)\cdot
%S_{[n]}
\sum_{j=1}^n\,
d\left(x_{i^*},p_j\right)
\right)\land\left(
\sum_{j=1}^k\,d\left(x_t,{\boldsymbol v}_j\right)
\le
\sum_{j=1}^k\,d\left(x_{i^*},{\boldsymbol v}_j\right)\right)
\right]\\
&\le&
\Pr\left[\exists i\in[h],\,
\left(
%S_{[n]}
\sum_{j=1}^n\,
d\left(x_i,p_j\right)>\left(1+\epsilon\right)\cdot
%S_{[n]}
\sum_{j=1}^n\,
d\left(x_{i^*},p_j\right)
\right)\land\left(
\sum_{j=1}^k\,d\left(x_i,{\boldsymbol v}_j\right)
\le
\sum_{j=1}^k\,d\left(x_{i^*},{\boldsymbol v}_j\right)\right)
\right]\\
&\le&
\sum_{i=1}^h\,
\Pr\left[
\left(
%S_{[n]}
\sum_{j=1}^n\,
d\left(x_i,p_j\right)>\left(1+\epsilon\right)\cdot
%S_{[n]}
\sum_{j=1}^n\,
d\left(x_{i^*},p_j\right)
\right)\land\left(
\sum_{j=1}^k\,d\left(x_i,{\boldsymbol v}_j\right)
\le
\sum_{j=1}^k\,d\left(x_{i^*},{\boldsymbol v}_j\right)\right)
\right]\\
%&\le&\Pr\left[\exists i\in B,\,
%\sum_{i=1}^h\,d\left(x_{i^*},{\boldsymbol v}_i\right)
%\ge \sum_{i=1}^h\,d\left(x_i,{\boldsymbol v}_i\right)
%\right]\\
%&\le&
%\sum_{i\in B}\,
%\Pr\left[
%\sum_{i=1}^h\,d\left(x_{i^*},{\boldsymbol v}_i\right)
%\ge \sum_{i=1}^h\,d\left(x_i,{\boldsymbol v}_i\right)
%\right]\\
&\stackrel{\text{Fact~\ref{Indykkeyfact}}}{<}&
\sum_{i=1}^h\,\exp{\left(-\frac{\epsilon^2 k}{64}\right)}\\
&=&
h\cdot \exp{\left(-\frac{\epsilon^2 k}{64}\right)}\\
&<&\epsilon,
\end{eqnarray*}
}% inequalities too long 20190901
%Consequently, taking
%$$
%h=\mathop{\mathrm{argmin}}_{i=1}^\ell\, \sum_{i=1}^k\,d\left(x_i,{\boldsymbol v}_i\right)
%$$
%works.
%Clearly, finding $h$ takes $O(h\ell)$ time.
where the second inequality uses $t\in[h]$.
\end{proof}

In short,
Lemma~\ref{wemaytakethebestaccordingtothesamples}
says how to find a
$((1+\epsilon)\kappa)$-approximate $1$-median
from $\{x_1,x_2,\ldots,x_h\}$
with probability greater than $1-\epsilon$, where $\kappa$ is the best
%if
approximation ratio among
$x_1$, $x_2$, $\ldots$, $x_h$.
%$\{x_1,x_2,\ldots,x_\ell\}$ contains a $\kappa$-approximate $1$-median, for all $\kappa>0$.
Note that computing $t$ in
Eq.~(\ref{thebestindexaccordingtorandomsamples})
requires no knowledge of the ordering $p_1$, $p_2$, $\ldots$, $p_n$.

\begin{figure}
\begin{algorithmic}[1]
\STATE $h\leftarrow \lceil 10^9(\log(1/\epsilon))/\epsilon\rceil$;
\STATE $k\leftarrow \lceil 10^9(\log(1/\epsilon))/\epsilon^2\rceil$;
\STATE Pick ${\boldsymbol u}_1$, ${\boldsymbol u}_2$, $\ldots$, ${\boldsymbol u}_h$,
${\boldsymbol v}_1$, ${\boldsymbol v}_2$, $\ldots$, ${\boldsymbol v}_k$
independently and uniformly at random from $[n]$;
\STATE $t\leftarrow \mathop{\mathrm{argmin}}_{i=1}^h\, \sum_{j=1}^k\,
d({\boldsymbol u}_i,{\boldsymbol v}_j)$, breaking ties arbitrarily;
\RETURN ${\boldsymbol u}_t$;
\end{algorithmic}
\caption{Algorithm {\sf approx.\ median} for {\sc ultrametric $1$-median}}
\label{mainalgorithmfor1median}
\end{figure}

\begin{lemma}\label{mainlemmafor1median}
%The algorithm
Algorithm {\sf approx.\ median}
in Fig.~\ref{mainalgorithmfor1median} outputs a
$((1+\epsilon)(1+2\epsilon))$-approximate $1$-median with probability greater than $1-2\epsilon$.
%, where
%$$\gamma
%=\left(1+\epsilon\right)\cdot\max\left\{1+C,\,\frac{1+\epsilon}{C}+\epsilon\cdot\left(
%1+\epsilon\right)\right\}.$$
\end{lemma}
\begin{proof}
With $h$
and
%$\{{\boldsymbol u}_i\}_{i=1}^h$
${\boldsymbol u}_1$, ${\boldsymbol u}_2$, $\ldots$,
${\boldsymbol u}_h$
as in
%line~3 of
{\sf approx.\ median},
\begin{eqnarray}
&&\Pr\left[\exists
%1\le i\le h,
i\in[h],
\, {\boldsymbol u}_i\in
\left\{p_1,p_2,\ldots,p_{\lceil\epsilon n\rceil}\right\}\right]\label{probabilityofpickinggoodfirst}\\
&=&1-\Pr\left[\forall
%1\le i\le h,
i\in[h],
\, {\boldsymbol u}_i\notin
\left\{p_1,p_2,\ldots,p_{\lceil\epsilon n\rceil}\right\}\right]\nonumber\\
&=&1-\left(1-\frac{\lceil\epsilon n\rceil}{n}\right)^h\nonumber\\
&>&1-\epsilon.\label{probabilityofpickinggoodlast}
\end{eqnarray}
%where the last inequality follows from line~1 of {\sf approx.\ median}.
When there exists $1\le i\le h$ satisfying
${\boldsymbol u}_i\in
\{p_1,p_2,\ldots,p_{\lceil\epsilon n\rceil}\}$,
Lemma~\ref{nowthisshouldbethekeylemma} asserts the existence
of a $(1+2\epsilon)$-approximate $1$-median in
$\{{\boldsymbol u}_1,{\boldsymbol u}_2,\ldots,{\boldsymbol u}_h\}$.
So
Eqs.~(\ref{probabilityofpickinggoodfirst})--(\ref{probabilityofpickinggoodlast})
force
$\{{\boldsymbol u}_1,{\boldsymbol u}_2,\ldots,{\boldsymbol u}_h\}$
to contain a $(1+2\epsilon)$-approximate $1$-median with probability
greater than $1-\epsilon$.
By
Lemma~\ref{wemaytakethebestaccordingtothesamples} (with $\{x_i\}_{i=1}^h$
substituted by $\{{\boldsymbol u}_i\}_{i=1}^h$),
%,
{\sf approx.\ median} outputs a $((1+\epsilon)\kappa)$-approximate $1$-median
with probability greater than $1-\epsilon$ if
$\{{\boldsymbol u}_1,{\boldsymbol u}_2,\ldots,{\boldsymbol u}_h\}$
contains a $\kappa$-approximate $1$-median,
%,
for all $\kappa>0$.
%where
Now take
$\kappa=1+2\epsilon$.
\end{proof}

\begin{theorem}\label{maintheorem}
{\sc Ultrametric $1$-median} has a Monte Carlo
$O((\log^2(1/\epsilon))/\epsilon^3)$-time
$(1+\epsilon)$-approximation algorithm
with success probability greater than $1-\epsilon$.
\end{theorem}
\begin{proof}
Invoke
%Lemma~\ref{yesoneoftherandompointsisgood}
%(with $\epsilon$ substituted by $\epsilon/100$)
%and Corollary~\ref{wemaytakethebestaccordingtothesamples} (with $\{x_i\}_{i=1}^\ell$
%substituted by $\{{\boldsymbol u}_i\}_{i=1}^\ell$).
Lemma~\ref{mainlemmafor1median} (with $\epsilon$ substituted by $\epsilon/4$)
%Then
%note that $1+C=1/C=1.618\ldots$
%Clearly, {\sf approx.\ median}
%in Fig.~\ref{mainalgorithmfor1median} runs in $O((\log^2(1/\epsilon))/\epsilon^3)$ time.
and calculate the running time of {\sf approx.\ median}.
\end{proof}

%It is easy to remove the $\mathop{\mathrm{polylog}}(1/\epsilon)$ term
%in the time complexity
%in Theorem~\ref{maintheorem} for the cost of a slightly lower success probability.
%But we do not know how to improve the $1/\epsilon^3$ part.

%\bibliographystyle{elsart-num}
%\bibliographystyle{alpha}
\bibliographystyle{plain}
\bibliography{ultrametric_1_median}

\noindent

\end{document}